\documentclass[12pt]{article}

\usepackage{amsthm}
\usepackage{amsmath}
\usepackage{thmtools,thm-restate}
\usepackage{amssymb}
\usepackage{graphicx}
\usepackage[margin=2cm, includefoot]{geometry}
\usepackage{commath}
\usepackage{xcolor}
\usepackage{color}
\usepackage{url}
\usepackage{algorithm}
\usepackage{algpseudocode}
\usepackage{bm}
\usepackage{verbatim}
\usepackage{subcaption}
\usepackage{enumitem}

\newcommand{\R}{\mathbb{R}}
\newcommand{\C}{\mathbb{C}}
\newcommand{\Z}{\mathbb{Z}}

\newcommand{\hM}{\hat{M}}
\newcommand{\x}{\hat{x}}
\newcommand{\z}{\hat{z}}
\newcommand{\p}{\hat{p}}
\newcommand{\q}{\hat{q}}

\newcommand{\E}{\mathbb{E}}

\newcommand{\N}{\mathcal{N}}

\newtheorem{thm}{Theorem}
\newtheorem*{thm*}{Theorem}
\numberwithin{equation}{section}
\numberwithin{thm}{section}

\newtheorem{cor}[thm]{Corollary}
\newtheorem{prop}[thm]{Proposition}
\newtheorem{lemma}[thm]{Lemma}

\newtheorem{remark}[thm]{Remark}
\usepackage{authblk}

\graphicspath{{./figures/}}

\begin{document}
\title{Dihedral multi-reference alignment}
\author{Tamir Bendory, Dan Edidin, William Leeb, and Nir Sharon}

\maketitle

\begin{abstract}
  We study the dihedral multi-reference alignment problem of estimating the orbit of a signal from multiple noisy observations of the signal, {acted on} by  random elements of the dihedral group. We show that  if the group elements are drawn from a generic distribution,  the orbit of a  generic signal is uniquely determined from the second moment of the observations. This implies that the optimal estimation rate in the high noise regime is proportional to the square of the variance of the noise.   
This is the first result of this type for multi-reference alignment over a non-abelian group with a non-uniform distribution of group elements. 
Based on tools from invariant theory and algebraic geometry, 
we also delineate conditions for unique orbit recovery for multi-reference alignment models over  finite groups (namely, when the dihedral group is replaced by a general finite group) when the group elements are drawn from a generic distribution.
Finally, we design and study numerically three computational frameworks for estimating the signal based on group synchronization, expectation-maximization, and the method of moments.

\end{abstract}


\section{Introduction}
We study the dihedral multi-reference alignment (MRA) model
\begin{equation} \label{eq:mra}
	y = g\cdot x + \varepsilon, \quad g\sim\rho, \quad \varepsilon\sim\N(0,\sigma^2I),
\end{equation}
where 
\begin{itemize}
	\item   $x\in\R^L$ is a fixed (deterministic) signal to be estimated; 
	\item $\rho$ is an unknown distribution  defined over the simplex $\Delta_{2L}$; 
	\item $g$ is a random element of the dihedral group $D_{2L}$, drawn i.i.d.\ from $\rho$, and acting on the signal by circular translation and reflection (see Figure~\ref{fig:example});   
	\item 	$\varepsilon$ is a normal isotropic i.i.d.\ noise  with zero mean and variance $\sigma^2$.
\end{itemize}
We wish to estimate the signal $x$ from  $n$ realizations (observations) of $y$,
\begin{equation} \label{eq:mra_observations}
	y_i =g_i\cdot x + \varepsilon_i, \quad i=1,\ldots,n,
\end{equation}
 while the corresponding group elements $g_1,\ldots,g_n$ are unknown. 
We note, however, that the signal can be identified only up to the action of an arbitrary element of the dihedral group. Therefore, unless a prior information on the signal is available, the goal  is estimating the orbit of signals  $\{g\cdot x | g\in D_{2L}\}$. This type of problem is often dubbed an orbit recovery problem.

The model~\eqref{eq:mra} is an instance of the more general MRA problem that was studied thoroughly in recent years~\cite{bandeira2014multireference,bendory2017bispectrum,bandeira2017estimation,abbe2017sample,abbe2018multireference,boumal2018heterogeneous,abbe2018estimation,perry2019sample,ma2019heterogeneous,bandeira2020optimal,bandeira2020non,romanov2021multi,abas2021generalized,hirn2019wavelet,aizenbud2021rank,katsevich2020likelihood,fan2020likelihood,ghosh2021multi,gao2019iterative,brunel2019learning,bendory2021compactification,bendory2021sparse}.
In its generalized version, the MRA model is formulated as~\eqref{eq:mra}, but  the signal $x$ may lie in an arbitrary vector space (not necessarily $\R^L$),  the dihedral group $D_{2L}$ is replaced by an arbitrary  {group $G$}, and $g\sim \rho$ is a distribution over $G$ (in some cases, an additional fixed linear operator acting on the signal is also considered, e.g.,~\cite{bandeira2020non, bandeira2017estimation, bendory2020super,bendory2020single}). The goal is to estimate the orbit of  $x$, under the action of the  group $G$. 

Most of the previous studies on MRA have considered the uniform (or Haar) distribution $\rho$ over the group elements. In particular, it was shown that in many cases, such as $x\in\R^L$ and a uniform distribution over $\Z_L$, the third moment suffices to recover a generic signal uniquely, {and consequently $n/\sigma^{6}\to\infty$ is a necessary condition for accurate estimation of generic signal ~\cite{bandeira2017estimation, kakarala2009completeness, perry2019sample, bendory2017bispectrum}.
In fact, this follows from a general result that in the low SNR regime $\sigma\to\infty$ (with a fixed dimension~$L$), a necessary condition for signal identification  is $n/\sigma^{2d}\to\infty$, where $d$ is the lowest order moment that identifies the orbit of signals uniquely~\cite{bandeira2017estimation, abbe2018estimation, bandeira2020optimal, perry2019sample} (see~\cite{romanov2021multi} for sample complexity analysis in high dimensions).}


The effect of non-uniform distribution on the sample complexity was first studied in~\cite{abbe2018multireference} for the  abelian group $\Z_L$ and $x\in\R^L$. It was shown that in this case the second moment suffices to identify the orbit of {generic signals} uniquely for almost any non-uniform distribution (rather than the third moment if the distribution is uniform). 
In this work, we extend~\cite{abbe2018multireference} for the non-abelian group $D_{2L}$ and show that for a generic distribution and signal, the second moment identifies the orbit of solutions. This implies that a necessary condition for accurate orbit recovery 
 under the model~\eqref{eq:mra} for $\sigma\to\infty$ and fixed $L$ is $n/\sigma^4\to\infty$. 
 This is the first result of this type for multi-reference alignment over a non-abelian group with a non-uniform distribution of group elements. {The fact that the group $D_{2L}$ is non-abelian makes the {analysis of the} orbit recovery problem significantly more difficult. The reason is that {this action of the dihedral group on $\R^N$} 
   cannot be diagonalized as we explain in Remark \ref{rem:diagonalizable}.
   It follows that there is no basis where the entries of the moment
   tensors are monomials. By contrast,  a previous work for the cyclic group
   $\Z_L$ took advantage of the fact that the entries of the moment tensors are monomials when expressed in the Fourier basis~\cite{abbe2018multireference}.}
The main theoretical results are summarized as follows.

\begin{figure}
	\begin{subfigure}[ht]{0.3\columnwidth}
		\centering
		\includegraphics[width=\columnwidth]{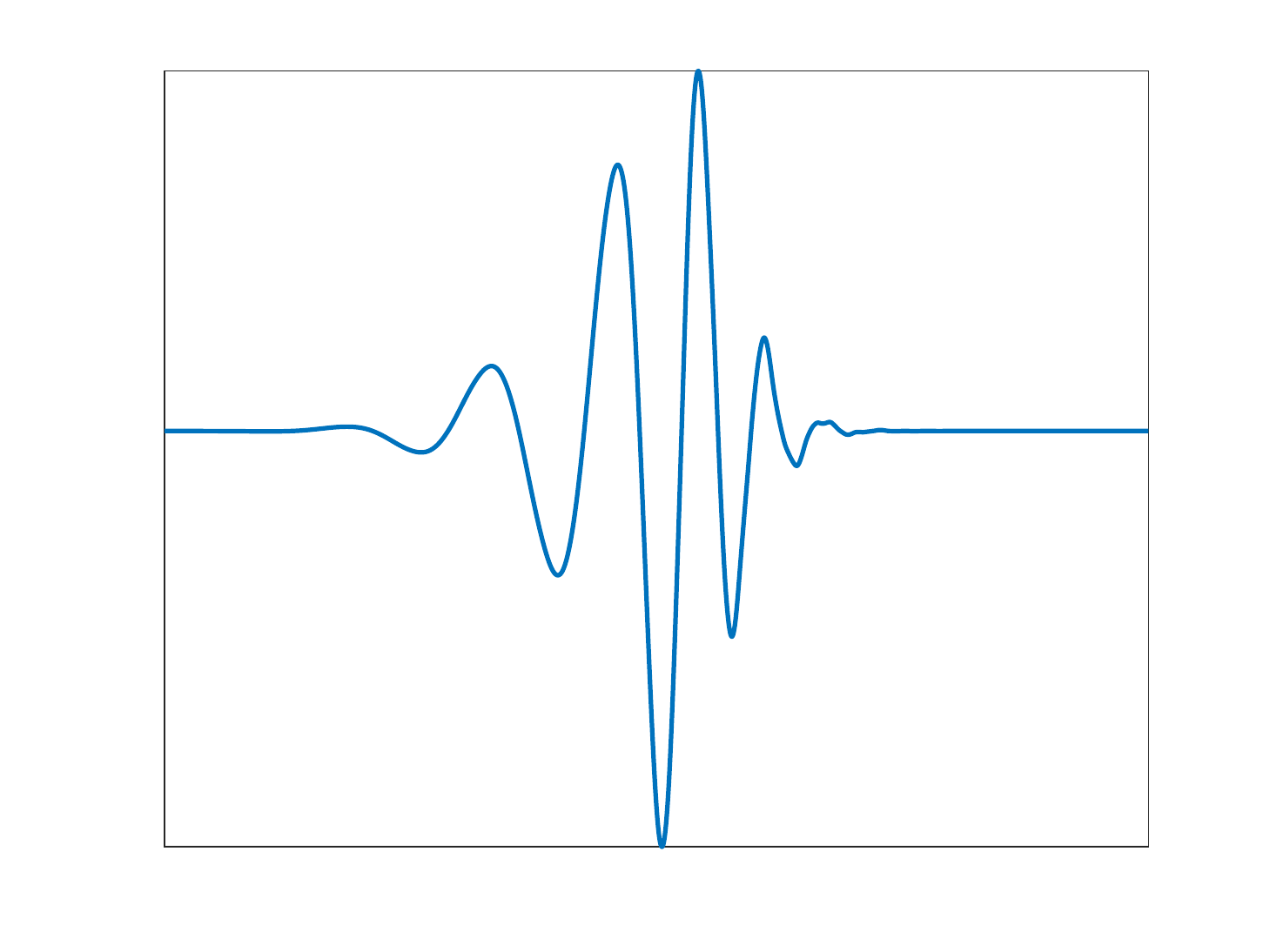}
		\caption{Signal}
	\end{subfigure}
	\hfill
	\begin{subfigure}[ht]{0.3\columnwidth}
	\centering
	\includegraphics[width=\columnwidth]{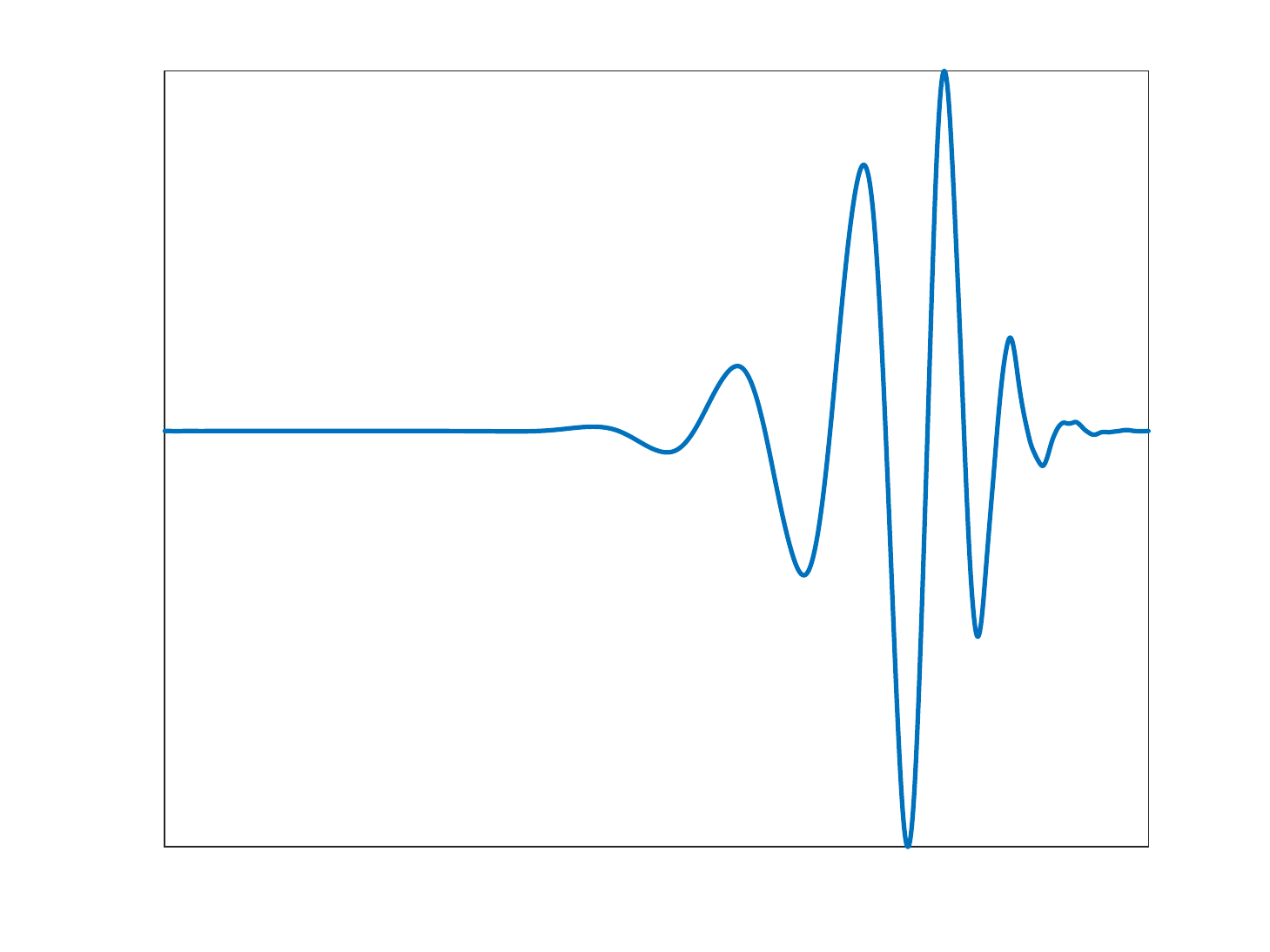}
	\caption{Shifted signal}
\end{subfigure}
	\hfill
\begin{subfigure}[ht]{0.3\columnwidth}
	\centering
	\includegraphics[width=\columnwidth]{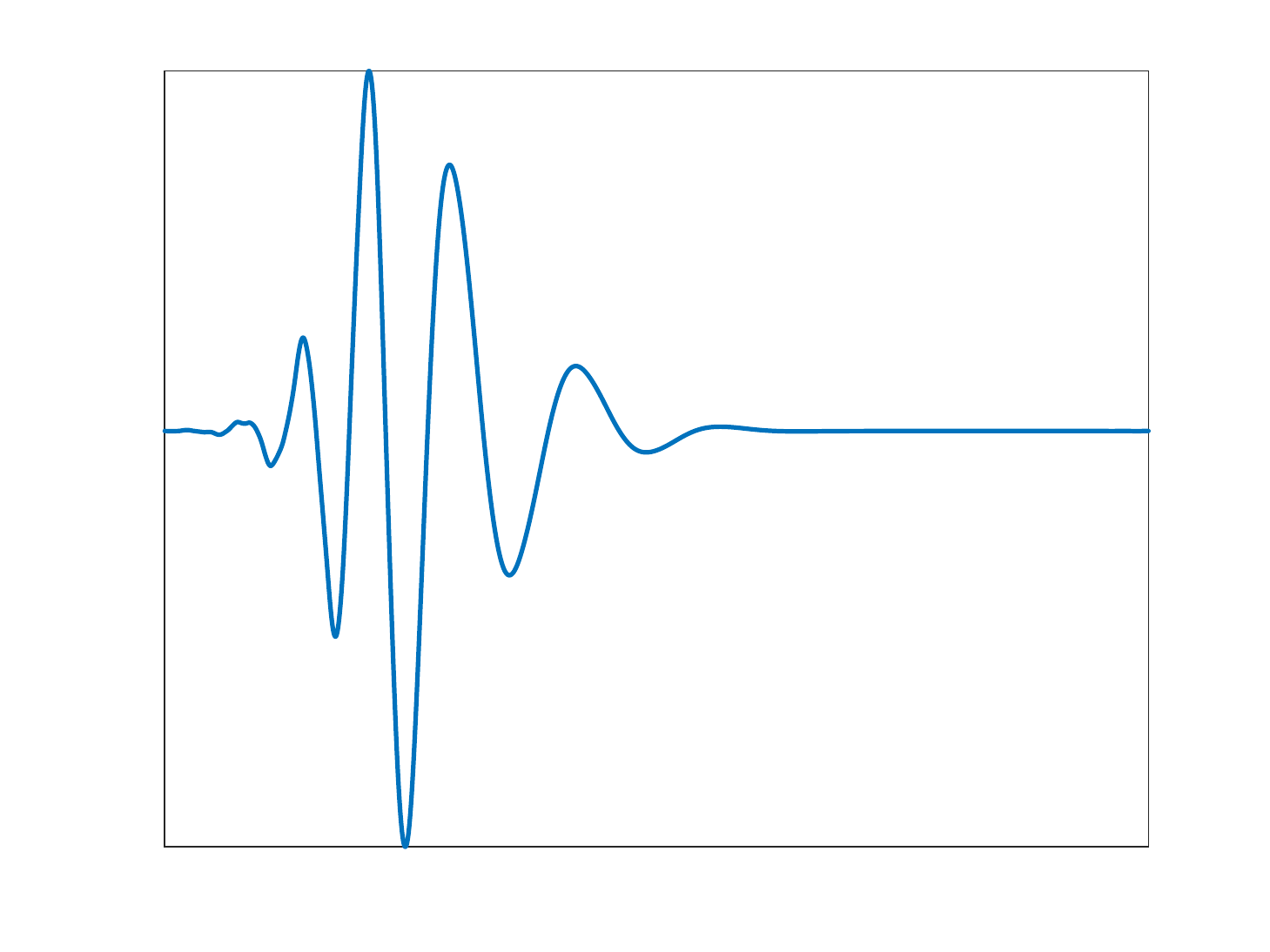}
	\caption{Shifted and reflected signal}
\end{subfigure}
	\caption{\label{fig:example} An example of the action of the dihedral group. The MRA problem~\eqref{eq:mra} entails estimating a signal, up to a global circular shift and reflection, from multiple noisy {copies of the signal}  acted upon by random elements of the dihedral group.}
\end{figure}

\begin{thm}[informal statement of the main theorem] \label{thm:main_informal}
  Consider the dihedral MRA problem~\eqref{eq:mra} with a generic probability distribution $\rho$. Then, the first and second order moments of $y$ are sufficient
  to {uniquely identify} almost all orbits.
\end{thm}
\begin{cor}[sample complexity] \label{cor:sample_complexity}
	 Consider the dihedral MRA problem~\eqref{eq:mra} in the low SNR regime $\sigma\to\infty$.
For a generic probability distribution  and a generic signal, 
$n/\sigma^4\to\infty$ is a necessary condition for accurate orbit {identification}.  
\end{cor}
Theorem~\ref{thm:main_informal}  is formulated in technical terms in Theorem~\ref{thm:main}, which is proved in Section~\ref{sec:uniqueness}. The proof is based on algebraic geometry tools {and is not constructive,  namely, it does not provide an explicit algorithm of how to recover the signal from the first and second moment.}
 {Section~\ref{sec:uniqueness} also discusses the precise meaning of the notion of generic signal and distribution.}
In Section~\ref{sec:general_theory}, we use  invariant theory to delineate  general conditions for orbit recovery from the second moment in general MRA models over finite groups.

The MRA model is mainly motivated by the molecular structure reconstruction problem in single-particle cryo-electron microscopy (cryo-EM)~\cite{bendory2020single}.
The aim of a  cryo-EM experiment is constituting a 3-D molecular structure from multiple observations. In each observation, the 3-D structure is acted upon by a random element of the non-abelian group of 3-D rotations SO(3).
In addition, the distribution over SO(3) is usually non-uniform and unknown~\cite{tan2017addressing,naydenova2017measuring,baldwin2020non,sharon2020method}. Therefore, this paper is an important step towards understanding the statistical properties and sample complexity of the cryo-EM problem
 
Section~\ref{sec:algorithms} introduces three statistical estimation frameworks to recover the orbit of $x$.  
The first framework is based on estimating the missing  group elements using
the {method of group synchronization}~\cite{singer2011angular,bandeira2020non}. Once the group elements were accurately estimated, estimating the signal can be obtained by aligning the observations and averaging out the noise. 
However, reliable estimation of  group elements is possible only if the noise level is low enough. To estimate the signal in high noise levels, we also suggest maximizing the marginalized maximum likelihood using expectation-maximization (EM). 
{The EM algorithm
	provides accurate estimations in a wide range of SNR regimes, although we have no theoretical guarantees to support it. Unfortunately, the computational burden of EM} rapidly increases with the number of observations $n$ and the noise level.  As a third method, we propose an estimator based on the method of moments, which works quite well in all SNRs and whose computational burden is roughly constant with the noise level and  moderately increases  with $n$.
 According to Theorem~\ref{thm:main_informal}, we only use  the first and second moments for the estimation. 
 {As with EM, characterizing the properties of the method of moments is left for  future research; see further discussion in Section~\ref{sec:algorithms}.} 

\section{Theory}

\subsection{The dihedral group} \label{sec:dihedral}
The dihedral group $D_{2L}$ is a group of order $2L$, which is usually defined as the group of symmetries of a regular
$L$-gon in $\R^2$. It is generated by a rotation $r$ of order $L$ corresponding to rotation by an angle $2\pi/L$ and
a reflection $s$ of order $2$. Since rotation does not commute with reflection, the group $D_{2L}$ is not abelian, but the relation $rs = sr^{-1}$ holds instead. {Since $r$ has order~$L$, $r^{-1}=r^{L-1}$.} 
The elements of $D_{2L}$ can be enumerated as $$\{1,r,\ldots, r^{L-1}, s, rs,\ldots, r^{L-1}s\},$$ where $1$ is the identity element. Note that the subset
$\{1,r, \ldots, r^{L-1}\}$ is a normal subgroup\footnote{A subgroup $H < G$ is normal
  if it is invariant under conjugation by elements of $G$.} isomorphic to the cyclic group $\Z_L$. MRA over the group $\Z_L$  was studied thoroughly, see for example ~\cite{bandeira2014multireference,bendory2017bispectrum,abbe2018multireference}.

{There are two natural ways to describe the action} of the dihedral group, {$D_{2L}$, on $\R^L$} one in the time (or spatial) domain
and one in the Fourier (frequency) domain.
Explicitly, in the time domain, the action of the dihedral group on a signal $x \in \R^L$ is given by 
\begin{equation}
	\begin{split}
	(r\cdot x)[\ell] &= x[(\ell-1)\bmod L], \\ 
	(s\cdot x)[\ell] &= x[-\ell \bmod L]. 	
	\end{split}
\end{equation}
Namely, $r$ cyclically shifts a signal by one entry, and $s$ reflects the signal.
The action of the dihedral group is illustrated in Figure~\ref{fig:example}.

If we apply the discrete Fourier transform to $\R^L$, then we can identify $\R^L$
with the real subspace of $\C^L$ consisting of $L$-tuples $(\x[0], \ldots , \x[L-1]) \in \C^L$ satisfying the condition $\overline{\x[\ell]} =
\x[-\ell \bmod L]$, where $\overline{\x[\ell]}$ is the conjugate of ${\x[\ell]}$.
In this case, the action of $D_{2L}$ is given by:
\begin{equation}
	\begin{split}
	(r\cdot \x)[\ell] &= e^{2 \pi \iota  \ell  /L} \x[\ell], \\ 
	(s\cdot \x)[\ell] &= \overline{\x[\ell]} = \x[-\ell \bmod L].
	\end{split}
\end{equation}
{\begin{remark} \label{rem:diagonalizable} We can see from this description that the action of $D_{2L}$ cannot be diagonalized for the following
    reason. If the action could be diagonalized, then there would have to be a basis for $\R^L$ which consists of simultaneous eigenvectors of 
    the rotation $r$ and the reflection $s$. However, the only eigenvector of the rotation $r$ which is also invariant under
  the action of the reflection $s$ is the vector $(1,0, \ldots, 0)$. 
  For a further reference, see \cite[p.~37]{serre1977linear}.
\end{remark}}
\subsection{Unique orbit recovery in dihedral MRA}
\label{sec:uniqueness}

We are now ready to present and prove the main result of this paper. 
Let $\rho\in \Delta_{2L}$ be a probability distribution on $D_{2L}$. We denote the probability of $r^k$ by {$p[k]$} and the probability of $r^ks$ by {$q[k]$}. 
Let ${p}$ and ${q}$ represent the vectors {$(p[0], \ldots , p[L-1])$}
and {$(q[0], \ldots , q[L-1])$}, respectively. Let $C_z\in\R^{L\times L}$ be a circulant matrix generated by $z\in \R^L$, 
{namely, the $i$-th column of $C_z$ is given by $z[(i-\ell)\bmod L]$ for $\ell=0,\ldots,L-1.$
Let } $D_z\in\R^{L\times L}$ be a diagonal matrix whose entries are~$z$. 
A direct calculation shows that the first two moments of the observations of~\eqref{eq:mra} are given by the following expressions (compare with~\cite{abbe2018multireference}).

\begin{lemma} \label{lem:moments}
	Consider the dihedral MRA model~\eqref{eq:mra}. 
	The first moment of $y$, $M^1\in\R^L$, is given by
	\begin{equation} 
		\E y := M^1(x,\rho) = C_x {p} + C_{sx} {q}= C_{{p}}x  + C_{{q}} sx. 
	\end{equation}
	The second moment of $y$, $M^2\in\R^{L\times L}$, is given by
	\begin{equation} \label{eq:sec_moment}
		\E yy^T :=M^2(x,\rho) = C_x D_{{p}} C_x^T + C_{sx} D_{{q}} C_{sx}^T+{\sigma^2I}.		
	\end{equation}
\end{lemma}

{Hereafter, we assume that the noise variance $\sigma^2$ is known, and thus the bias term $\sigma^2I$ can be removed. Indeed, the variance of the average of each observation $\frac{1}{\sqrt{L}}\sum_{\ell=0}^{L-1}y_i[\ell]$ (which is invariant under the group action) is an unbiased estimator of the noise variance, and is consistent as $\sigma^4/n \to 0$. We also remark that in many applications, including cryo-EM, the noise level can often be readily estimated from the data~\cite{bendory2020single}.}


To present the main result of this paper, it will be convenient to consider the Fourier counterpart of the moments, defined by
 \begin{align*}
	\hM^1 &= \E Fy = FM^1(x,\rho), \\
	\hM^2 &= \E Fy
	(Fy)^*= FM^2(x, \rho)F^{*},
\end{align*}
where $F\in\C^{L\times L}$ is the discrete Fourier transform (DFT) matrix.   

We say that a condition holds for
 generic signals (or distributions) if the set of signals (distributions) for which the condition
does not hold is defined by polynomial conditions. {The precise meaning of generic signals, in the context of this work, is discussed at the end of this section.}

The main result of this paper is as follows.
\begin{thm}[Orbit recovery] \label{thm:main}
  For generic signal $x$ and generic distribution $\rho$, the $D_{2L}$ orbit of~$x$ is
  uniquely determined   
  by $\hM^1[0]= \x[0]$  and at most $\sim 2.5L$ entries of the matrix $\hM^2$.
More precisely, there exist non-zero polynomials $Q_1,\dots,Q_r$ such that if $Q_1(x,\rho),\dots,Q_r(x,\rho)$ are not all zero, then for any $(z,\rho')$ with $M^1(z,\rho') = M^1(x,\rho)$ and $M^2(z,\rho') = M^2(x,\rho)$, $z$ is in the same~$D_{2L}$ orbit as $x$.
\end{thm}

\begin{remark} Our method of proof necessarily requires that all of the entries
    of $Fx$ are non-zero, where $F$ is the discrete Fourier transform matrix; similar assumptions are often stated in the MRA literature, see for example~\cite{bendory2017bispectrum,abbe2017sample,perry2019sample,bandeira2020optimal}.
     However, our proof also requires that additional, less explicit, polynomials in the entries of $x,\rho$ be non-vanishing. This is discussed at the end of the proof.
  \end{remark}

\begin{proof}
	Let us define
 \begin{align*}
 { \hat{ p}} &= (\p[0],\ldots , \p[L-1]),\\ 
{\hat{ q}} &= (\q[0], \ldots , \q[L-1]), \\
\hat x &= (\x[0], \ldots , \x[L-1]).
\end{align*}  
Note that the second moment in Fourier domain can be written as
  \begin{equation}
\hM^2 = \frac{1}{L}\left(D_{\hat x} C_{\hat{ p}} D_{\overline{\hat x}} +
D_{\overline{\hat{ x}}} C_{\hat{q}}D_{\hat{ x}}\right).
  \end{equation}
  Moreover, since ${ p}, { q}, x$ are real, we have {the symmetry relations}
  \begin{align*}
\x[L-i] &= \overline{x[i]}, \\
\p[L-i] &= \overline{\p[i]}, \\
\q[L-i] &= \overline{\q[i]}.
  \end{align*}

  Define  
  \begin{equation} \label{eq.Mij} M_{i,j} = \p[i+j] \x[i] \x[j] + \q[L-i-j]\x[L-i] \x[L-j],
    \end{equation}
  so
  $M_{i,j}$ is $L \hM^2[i,L-j]$. Our goal is to show that
  knowledge of $\hM^1[0]$ and $O(L)$ of the entries $M_{i,j}$ determine the
  orbit of $x$.

 Since $\rho$ is a probability distribution, we note that
  \begin{equation}
  \p[0] + \q[0] = p[0] + \ldots + p[L-1] + q[0] + \ldots + q[L-1] = 1.
  \end{equation}
  Thus, $M_{i,-i} = |\x[i]|^2$. It follows that knowledge of $M^2(x, \rho)$
  determines the power
  spectrum of~$x$. Replacing $x$ by the vector whose Fourier transform
  {has entries $\x[i]/|\x[i]|$}, we may assume that each~$\x[i]$ lies on the unit circle. Since
  $\x[0]$ is real, we take $\x[0] = 1$.
  With this assumption, the formula for $M_{i,j}$ can be written as 
  \begin{equation}
    M_{i,j} = \p[i+j] \x[i] \x[j] + \q[L-i-j]/(\x[i] \x[j]).
    \end{equation}
  
Given a vector $x$ and distribution $\rho$, consider the set $I$ of vectors $z \in \R^L$ such that $M^1(z, \rho') = M^1(x,\rho)$ and $M^2(z, \rho') = M^2(x,\rho)$ for some probability distribution $\rho'$
  on the dihedral group $D_{2L}$.
  We will show that for generic $(x, \rho)$ there are only
  $2L$  possible $z$'s in this set.
Note that the distribution is uniquely determined by the signal $z$, because the moments are linear functions of the distribution.
  Since the $D_{2L}$ orbit of $x$ is contained
  in the set $I$, we conclude that the orbit of $x$ is determined by
  the moments of degree one and two.

  Determining that the set $I$ consists of at most $2L$ vectors is equivalent to showing that the following system of 
  equations has at most $2L$ solutions:
    \begin{equation} \label{eq.thesystem}
  	\p'[i+j] \z[i] \z[j] + \q'[L-i-j]/(\z[i] \z[j]) = M_{i,j},
  \end{equation}
   where $|\z[i]| = 1$, $\z$ is the Fourier transform of a vector in $\R^L$, and ${\hat{\rho}'} = (\p', \q')$
  is the Fourier transform of a probability distribution on $D_{2L}$.
	Consider the equations 
	\begin{equation} \label{eq.clutch}
	  \p'[1] \z[\ell] \z[1 -\ell] + \q'[L-1]/(\z[\ell] {\z[1 -\ell])} =
          M_{\ell, 1 -\ell}, \quad \ell=0,\ldots,L-1,
	\end{equation}
        where the indices are taken modulo $L$.
	For each fixed $\ell$, we can view equation~\eqref{eq.clutch} as a linear
	equation in $\p'[1], \q'[L-1]$.
	For the system to  have a solution, it must be consistent.
	Taking the pair of equations when $\ell = 1$ and $\ell = m+1$ with $m \geq 1$, we obtain 
	\begin{equation} \label{eq.n0}
	  {\q'[L-1]} = 
             {{\left(\z[1] \z[m+1]^2 M_{1,0} - \z[1]^2 \z[m]\z[m+1] M_{m+1,-m}
                 \right)}\over{\z[m+1]^2 - \z[1]^2\z[m]^2}}.
	\end{equation}
	Equating equation~\eqref{eq.n0} with $m =1$ and $m=n+1$, 
        we see that $\z[n+1]$ satisfies the 
        following quadratic equation in terms of $\z[1], \z[2], \z[n]$:
	\begin{equation} \label{eq.quadratic10}
		\begin{split}
		  (M^2_{2,-1} \z[1] \z[2] - M_{1,0} \z[1]^3) \z[n+1]^2&+ M^2_{n+1,-n}(\z[1]^4\z[n] -\z[2]^2\z[n])\z[n+1] + \\&  M^2_{1,0}\z[1]\z^2[2]\z[n]^2 -
                  M^2_{2,-1} \z[1]^3\z[2] \z^2[n] = 0.
		\end{split}
        \end{equation}
       Note that expressions of the form $M^2_{i,j}$ refer to exponents in this formula.

        If $n > 2$, the three equations from~\eqref{eq.thesystem}  with
        $(i,j) = (n+1,0), (n,1), (n-1,2)$, respectively,
        yield three linear equations for $\p'[n+1], \q'[L-n-1]$ whose
        coefficients are rational expressions in $\z[1], \z[2],\z[n-1],\z[n],
        \z[n+1]$.
        The same analysis as above shows that {$\z[n+1]$}
        satisfies an additional
        quadratic equation in $\z[1],\z[2],\z[n-1],\z[n]$: 
        \begin{equation} \label{eq.quadraticnplusone}
        	\begin{split}
&          (M^2_{n,1} \z[1]\z[n] - M^2_{n-1,2} \z[2]\z[n-1])\z[n+1]^2+ M_{n+1,0}(\z[2]^2\z[n-1]^2 - \z[1]^2\z[n]^2)\z[n+1] \\ 
& + (M^2_{n-1,2}
\z[1]^2\z[2]\z[n-1]\z[n]^2 - M^2_{n,1} \z[1]\z[2]^2 \z[n-1]^2 \z[n]) = 0.
        	\end{split}
        \end{equation}
        Since $\z[n+1]$ satisfies the two non-equivalent quadratic equations
        \eqref{eq.quadratic10} and \eqref{eq.quadraticnplusone}, we can solve for $\z[n+1]$ in terms of $\z[1],\z[2]$ and $\z[n]$ and we obtain the following
        expression for $\z[n+1]$ as a rational function of $\z[1], \z[2], \z[n-1], \z[n]$:
\begin{equation}  \label{eq.znplusonerational}
	\z[n+1] = \frac{a}{b},
\end{equation}        
where 
\begin{equation} \label{eq.numerator}
	\begin{split}
	a &= M^2_{1,0}M^2_{n,1}(\z[1]^2\z[2]^2\z[n]^3 -
\z[1]^4\z[2]^2\z[n-1]^2\z_n)  + 
M^2_{2,-1}M^2_{n,1}(\z[1]^2\z[2]^3\z[n-1]^2 \z[n] -
\z[1]^4\z[2]\z[n]^3) \\ &+ M^2_{1,0}M^2_{n-1,2}(\z[1]^5\z[2]\z[n-1]
\z[n]^2 -\z[1]\z[2]^3 \z[n-1]\z[n]^2),
	\end{split}
\end{equation}
and 
\begin{equation} \label{eq.denominator}
	\begin{split}
	b&= M_{n-1,1}M_{n+1,-n}(\z[1]\z[2]^2\z[n]^2 -\z[1]^5\z[n]^2) +
M^2_{n-1,2}M^2_{n+1,-n}(\z[1]^4\z[2]\z[n-1]\z[n] -
\z[2]^3\z[n-1]\z[n]) \\ &+  M^2_{n+1,0}M^2_{2,-1}(\z[1]\z[2]^3\z[n-1]^2
-\z[1]^3\z[2]\z[n]^2) + M^2_{n+1,0}M^2_{1,0}(\z[1]^5\z[n]^2 -
\z[1]^3\z[2]^2\z[n-1]^2).	
\end{split}
\end{equation}
        When $n =2$, the equations in \eqref{eq.thesystem} corresponding to $(i,j) = (2,1)$ and
        $(i,j) =(1,2)$ are identical so we need another method to express
        $\z[3]$ as a rational function of $\z[1],\z[2]$.
        To get a second quadratic equation in this case, consider the equations of \eqref{eq.thesystem} corresponding to the pairs $(2,0), (1,1), (3,-1)$ to obtain
the quadratic equation
        \begin{equation} \label{eq.quadratic32}
          (M_{1,1}^2 \z[1]^2 - M^2_{2,0} \z[2]) \z[3]^2 + M^2_{3,-1} \z[1]
          (\z[2]^2 -\z[1]^4) + \z[1]^4\z[2](M^2_{2,0}\z[1]^2 - M^2_{1,1}\z[2]) = 0.
        \end{equation}
        We then obtain the following expression for $\z[3]$ as a rational
        function of $\z[1]$ and $\z[2]$:
        \begin{equation} \label{eq.z3rational}
          {{\z[1]\z[2](M^2_{1,0}M^2_{2,0} \z[1]^4 - M_{1,0}M^2_{1,1}\z[1]^2\z[2]
              - M_{2,0}M_{2,-1}\z[1]^2\z[2] + M_{1,0}M_{2,0}\z[2]^2)}
            \over{(M_{1,0}M_{3,-1}\z[1]^4 - M_{2,-1}M_{3,-1}\z[1]^2\z[2] -
              M_{1,1}M_{3,-2}\z[1]^2\z[2] + M_{2,0}M_{3,-2} \z[2]^2)}}.
        \end{equation}

At this point we have shown that knowledge of $\z[1], \z[2]$ determine
        $\z[n+1]$ for $n \geq 2$, assuming that the rational expressions
        \eqref{eq.znplusonerational} and \eqref{eq.z3rational} are well defined {(see the discussion at the end of the proof)}.
        We can also use the quadratic equations \eqref{eq.quadratic10} (with $n=3$) and \eqref{eq.quadratic32} to obtain a second expression for $\z[3]^2$
        as {a} rational function of
        $\z[1]$ and $\z[2]$.
        Equating this expression for $\z[3]^2$ with the square
        of the expression for $\z[3]$ given by \eqref{eq.z3rational},
        we obtain the following palindromic quartic equation for $\z[2]$ in terms
        of $\z[1]$:
\begin{equation} \label{eq.quartic}
  A_0 \z[1]^8 + A_1\z[1]^6 \z[2] + A_2\z[1]^4 \z[2]^2 + A_1\z[1]^2\z[2]^3
  + A_0 \z[1]^4 =0,
\end{equation}
where 
\begin{eqnarray*}
  A_0 & = & M_{1, 0}^2 M_{2, 0}^2 - M_{1, 0} M_{2, 0} M_{3, -2} M_{3, -1}\\
  A_1 & = & -2 M_{1, 0}^2 M_{1, 1} M_{2, 0} - 2 M_{1, 0} M_{2, -1} M_{2, 0}^2 + 
  M_{1, 1} M_{2, 0} M_{3, -2}^2\\
  & & + M_{1, 0} M_{1, 1} M_{3, -2} M_{3, -1}+
 M_{2, -1} M_{2, 0} M_{3, -2} M_{3, -1} + M_{1, 0} M_{2, -1} M_{3, -1}^2\\
 A_2 & = & M_{1, 0}^2 M_{1, 1}^2 + 2 M_{1, 0} M_{1, 1} M_{2, -1} M_{2, 0} + 
 2 M_{1, 0}^2 M_{2, 0}^2 + M_{2, -1}^2 M_{2, 0}^2\\
 & & - M_{1, 1}^2 M_{3, -2}^2 -
  M_{2, 0}^2 M_{3, -2}^2 - 2 M_{1, 1} M_{2, -1} M_{3, -2} M_{3, -1} - 
  M_{1, 0}^2 M_{3, -1}^2 - M_{2, -1}^2 M_{3, -1}^2.
\end{eqnarray*}

Taking the complex conjugate of \eqref{eq.quartic} and using the fact {that}
$\z[i]$ lies on the unit circle so $\overline{\z[i]} = \z[i]^{-1}$, we obtain  
\begin{equation} \label{eq.quarticinv}
  \overline{A_0} \z[1]^{-8} + \overline{A_1}\z[1]^{-6}\z[2]^{-1} +
  \overline{A_2}\z[1]^{-4}\z[2]^{-2} + \overline{A_1}\z[1]^{-2}\z[2]^{-3} +
  \overline{A_0} \z[1]^{-4}{=0}.
\end{equation}
Multiplying~\eqref{eq.quarticinv} by $\z[1]^8\z[2]^4$, we obtain a second quartic equation satisfied by $\z[2]$:
\begin{equation} \label{eq.quarticconj}
  \overline{A_0} \z[1]^{8} + \overline{A_1}\z[1]^{6}\z[2] +
  \overline{A_2}\z[1]^{4}\z[2]^{2} +
  \overline{A_1}\z[1]^{2}\z[2]^{3} + \overline{A_0} \z[2]^{4}{=0}.
\end{equation}
Now take $\sqrt{-1}(
A_0\eqref{eq.quarticconj} -\overline{A_0} \eqref{eq.quartic} )$
and we obtain the following equation with real coefficients
\begin{equation}
  \z[1]^2\z[2]( B_1\z[1]^4 + B_2\z[1]^2\z[2] + B_1\z[2]^2) =0,
\end{equation}
where $B_1 = 2 \Im(\overline{A_0}A_1)$ and $B_2 = 2 \Im(\overline{A_0}A_2)$ ($\Im$ stands for the imaginary part of a complex number).
Since $\z[1],\z[2] \neq 0$, we see that $\z[2]$ satisfies the real palindromic
equation
\begin{equation} \label{eq.final}
B_1 \z[1]^4 + B_2 \z[1]^2\z[2] + B_1 \z[2]^2 = 0.
\end{equation}
Since the equation is palindromic, if $\z[2]$ is a root then
$1/\z[2] = \overline{\z[2]}$ is also necessarily a root.

At this point we have shown that given $\z[1]$, there are (at most) two possible values
for $\z[2]$ provided that $B_1, B_2$ are non-zero. Once we have $\z[1], \z[2]$, the values of
$\z[3], \ldots , \z[L/2]$ are uniquely determined, assuming that the rational expressions~\eqref{eq.znplusonerational} and \eqref{eq.z3rational} are well-defined.
However, we have no constraints on $\z[1]$ other than it lies on
the unit circle. Indeed, 
{the polynomial equations \eqref{eq.quadratic10}, \eqref{eq.quadraticnplusone}, \eqref{eq.quadratic32}, \eqref{eq.final}}
are weighted homogeneous where the variable $\z[n]$ has weight~$n$. In other words, 
if $(\z[1], \z[2], \z[3], \ldots , \z[L/2])$ is a solution, then
$(\lambda \z[1], \lambda^2 \z[2], \ldots , \lambda^{L/2} \z[L/2])$
{will be a solution for any $\lambda \in S^1$.}
When $L$ is even we obtain a constraint on $\z[1]$ by noting that
$\z[L/2]^2 =1$ since $z[L/2] = 1/z[L/2]$ because $L/2 = L-L/2$. Hence we must have $(\lambda^{L/2})^2 =1$;
i.e., $\lambda^L =1$, so $\lambda$ is an $L$-th root of unity. Hence our system can have at most $2L$ solutions.
When $L$ is odd,
{we observe that $\z[(L-1)/2] = \overline{\z[L-(L-1)/2]} = \overline{\z[(L+1)/2]} = \z[(L+1)/2]^{-1}$, and so $\z[(L-1)/2] \z[(L+1)/2] = 1$; and replacing $\z[k]$ with $\lambda^k\z[k]$, we find $\lambda^L\z[(L-1)/2] \z[(L+1)/2] = 1$, i.e.,  $\lambda^L = 1$. Hence, our system can only have at most $2L$ solutions in this case as well.
}

{{\bf Generic Conditions.} To complete the proof, we explain why for generic
  $(x,\rho)$ with all $\hat{x}[i]$ non-zero, the quadratic
  equation~\eqref{eq.final} is non-zero and the rational expressions~\eqref{eq.znplusonerational} and \eqref{eq.z3rational} are well-defined. To show that~\eqref{eq.final} is non-vanishing for generic $(x,\rho)$ we must show that
  $\overline{A_0}A_1$ and
  $\overline{A_0}A_2$ are not pure real. This is a real polynomial condition on $A_0,A_1, A_2$, which are themselves polynomials in the entries of $\rho$
  and $x$. To prove that this condition holds generically, it suffices to prove that this is the case for a single choice of $(x,\rho)$. {Moreover, since
    the simplex is Zariski dense in the linear subspace $\sum p[i] + q[i] = 1$,
    it suffices to verify this when the vector $\rho$ lies in this subspace
    without necessarily being a {probability} distribution. Applying
  the Fourier transform, it suffices to verify that the condition
  holds for a single pair $(\hat{x}, \hat{p})$ with $\hat{p}[0] + \hat{q}[0] =1$.
  The expressions for $A_0, A_1,A_2$
  are determined by the moment entries $M_{1,0}, M_{2,-1}, M_{3,-2}, M_{2,0}, M_{1,1}, M_{3,-2}$, which are in turn determined by the seven values $\hat{x}[1], \hat{x}[2], \hat{x}[3],
  \hat{p}[1], \hat{p}[2], \hat{q}[L-1], \hat{q}[L-2]$. In particular
  if we set $\{x[1],x[2],x[3]\} = \{1,1,\sqrt{-1}\}$,
  $\{\hat{p}[1], \hat{p}[2]\} = \{1,1\}$ and $\hat\{q[L-1],q[L-2]\} = \{1+\sqrt{-1}, \sqrt{-1}\}$, 
  then $B_1 = 16$ and $B_{2} = -32$.}

  Since $(z,\rho') = (x, \rho)$ automatically satisfies the system of
  equations~\eqref{eq.thesystem}, it follows that $\z[1] = {\lambda} x[1]$,
  where ${\lambda}$ is an $L$-th root of unity. Moreover, we know that
  that when $\z[1] = \x[1]$, the quadratic equation~\eqref{eq.final} has solutions $\z[2] = \{\x[2], 1/\x[2]\}$. Hence, if 
  $\z[1] = {\lambda} \x[1]$,  then \eqref{eq.final} has solutions $\z[2] = \{{\lambda}^2 \x[2], 1/({\lambda}^2 \x[2])\}$. It follows that the rational expression~\eqref{eq.z3rational} is well-defined as long the polynomial expressions
$$  (M^2_{1,0}M^2_{2,0} \z[1]^4 - M_{1,0}M^2_{1,1}\z[1]^2\z[2]
  - M_{2,0}M_{2,-1}\z[1]^2\z[2] + M_{1,0}M_{2,0}\z[2]^2),$$
  and
  $$(M_{1,0}M_{3,-1}\z[1]^4 - M_{2,-1}M_{3,-1}\z[1]^2\z[2] -
  M_{1,1}M_{3,-2}\z[1]^2\z[2] + M_{2,0}M_{3,-2} \z[2]^2),$$
  are both non-zero when $\{\z[1],\z[2]\} = \{{\lambda} \x[1], {\lambda}^2 \x[2]\}$
  or $\{\z[1],\z[2]\} = \{ {\lambda} \x[1], 1/({\lambda}^2 \x[2])\}$. If this is the case, then it follows that $\z[3] = {\lambda}^3 \x[3]$ or $\z[3] = 1/({\lambda}^3 \x[3])$ because we know that $(\x[1], {\lambda} \x[2], {\lambda}^3 \x[3])$ and $(\x[1], 1/({\lambda}^2 \x[2]),
  1/({\lambda}^2 \x[3]))$ are the first three entries of a vector in the $D_{2N}$ orbit
  of the vector $\x$. {Using \eqref{eq.numerator} and \eqref{eq.denominator}, we} can now continue recursively to obtain sufficient genericity conditions on the pair $(x,\rho)$.
}
\end{proof}
  \begin{remark}
    As can be seen from the proof, we only use $\sim 5L/2$ of
    the entries of $M_{i,j}$ (out of $L^2$ entries overall) to determine the orbit of $x$. Precisely, we only use
    the $\sim 5L/2$ entries  $M_{\ell,-\ell},  M_{\ell, 1-\ell},M_{\ell+1,0}, M_{\ell,1} $ and $ M_{\ell-1,2}$ for  $\ell=0,\ldots,L/2$.
     A similar observation was made in~\cite{bendory2021signal}. 
  \end{remark}

{
\subsection{General theory for MRA  with a general distribution over finite groups}
\label{sec:general_theory}
The purpose of this section is to discuss the theory of moments
for the MRA problem for finite groups. Our goal is to highlight the mathematical differences between uniform and generic distributions on the group $G$.
Precisely, the dihedral MRA model~\eqref{eq:mra} we consider here is a special
case of the following MRA problem: 

Recover  a signal $x \in {V}$ from moment measurements of $g_i \cdot
x + \epsilon_i$, where the group elements $g_i$ are chosen `at random'
from a finite group~$G$ and~${V}$ is a
finite dimensional vector space.\\

\subsubsection{Uniform distribution} 
The case of a uniform distribution of  
the group elements $g \in G$ was studied in depth in~\cite{bandeira2017estimation}. 
For the uniform distribution, the $n$-th moment 
$$M^{n}= {1\over |G|} \sum_{g\in G} gx^{\otimes n},$$
is a tensor whose components generate the vector space
of invariant polynomial functions of degree $n$ on $V$.
An important theoretical result {whose proof uses Jennrich's algorithm for decomposing a three-tensor} is the following
theorem:
\begin{thm*} \cite[Theorem D.2]{bandeira2017estimation}
  Let $G$ be a finite group and let $V$ be the regular representation
  of $G$ over $\R$, then the generic orbit $Gx$ consists of linearly independent vectors and consequently generic recovery is possible from degree 3 invariants.
\end{thm*}
(The regular representation of a finite group
is the $|G|$ dimensional vector space of functions $G \to \R$ where the group $G$ acts by $(g\circ f)(h) = f(g^{-1}h)$.)

Since $\R^L$ is the regular representation of the cyclic group $\Z_L$, the Theorem above implies that for the uniform distribution
on $\Z_L$ the generic vector $x \in \R^L$ can be recovered from the third order moment: a result originally proved in \cite{bendory2017bispectrum,perry2019sample}.
Note, however, that this result cannot be applied
for the action of $D_{2L}$ on $\R^{L}$ because $\R^L$ is not the regular representation of $D_{2L}$ since its dimension is smaller than the order of the group $D_{2L}$.
As a result, we do not know if the first three moments suffice to recover
a generic orbit when the distribution in $D_{2L}$ is uniform.

\subsubsection{Generic distributions}
We now give a theoretical analysis of the situation where
the group elements $g_i$ are taken from a {\em
  generic distribution} on the finite group  $G$,
as we do here for the dihedral group $D_{2L}$ and as was done in \cite{abbe2018multireference} for the cyclic group $\Z_L$.

Observe that a probability distribution on a finite group
is a function $\rho\colon G \to \R$ satisfying the conditions
$\rho(g) \geq 0$ for all $g \in G$ and $\sum_{g \in G} \rho(g) = 1$.
Thus, a probability distribution is a vector $\rho$ in the regular representation which lies in the simplex $\Delta_{|G|} \subset R(G)$, where $R(G)$ denotes
the regular representation.  
By definition, the $n$-th order moment associated to a probability distribution~$\rho$ on~$G$, 
$M^n := \sum_{g \in G} \rho(g) (gx)^{\otimes n}$
is a $n$-tensor of
invariant polynomials
of bidegree $(1,n)$ on $R(G) \times V$. 
Of particular interest in this paper is the second order moment
$M^2(x,\rho) = \sum_{g \in G}  \rho(g) (gx) (gx)^T$, when $G$ is the dihedral group.
In this case, the second order moment gives a collection of invariant functions of total
degree $3$ on $R(G) \times V$.


The following result which is of purely theoretical interest
states that the
 orbit of a generic pair $(\rho,x) \in R(G)$ can be determined from the full
 collection of degree $3$ invariant polynomials.
\begin{prop} \label{prop.generic_distribution}
	The set of all degree 3 invariants on $R(G) \times V$ determines the $G$-orbit
	of a generic pair $(\rho,x)$.
\end{prop} 
\begin{proof}
  As in \cite{bandeira2017estimation} it suffices to show that the orbit of a generic $(\rho, x) \in R(G) \times V$ consists of linearly independent vectors.
  Note that projection map $R(G) \times V \to R(G)$ is $G$-invariant. Thus
  the projection of the orbit $G(\rho, x)$ to $R(G)$ is the $G$-orbit of
  $\rho$ in $R(G)$. It then follows from \cite[Theorem D.2]{bandeira2017estimation} that $G \rho$ consists of linearly independent vectors
  and hence so does $G(\rho, x)$. We can then recover the orbit
  from degree three invariants. 
\end{proof}
\begin{remark} Note that there is no way to estimate all of the degree invariants
  in $R(G) \times V$ from a given set of MRA measurements. For this reason, Proposition~\ref{prop.generic_distribution} is only of theoretical interests. In particular, note that even from a theoretical point of view our results for the dihedral
  group acting on $V = \R^L$
  are much stronger that the guarantee given by Proposition~\ref{prop.generic_distribution} since they state that quite a small subset of the degree three invariants of $R(G) \times V$ are sufficient to recover generic orbits.
\end{remark}

\paragraph{List recovery.}
Following the terminology of \cite[Section 1.4]{bandeira2017estimation},
we say that a signal $x$ admits {\em list recovery} from a set of moment
measurements if there are a finite number of orbits with same moments.
As was done in \cite[Section 4.2.2]{bandeira2017estimation}, one can
use the Jacobian criterion to determine if a collection of MRA moments
with generic distribution allows list recovery for a generic orbit $x$.

Precisely, let  $f_1, \ldots , f_r \in \R[p_1, \ldots , p_g,x_1, \ldots , x_L]^G$ be a collection of invariant
polynomials of degrees $(1,d_1), \ldots , (1,d_r)$ corresponding to some set
of entries of the moment tensors\\ $M^{d_1}(\rho, x), \ldots , M^{d_r}(\rho, x)$.
Then these moments are sufficient to allow list recovery of a generic signal
if and only if the rank of the Jacobian matrix $J(f_1, \ldots , f_r)$ equals $|G| -1 + L$.
The rank of the Jacobian can be effectively computed in examples, but this will be considered in another work.

An easy consequence of the Jacobian criterion is the following corollary.
\begin{cor}
	If $|G| > {L+1\choose{2}} +L $,  then list recovery is impossible from second order moments.
\end{cor}
\begin{proof}
  Since the second order moment tensor is symmetric, the total number
  of first and second order moments is 
$L + {L+1\choose{2}}$ which is smaller than $G-1+L$ so list recovery is impossible.
\end{proof}

\paragraph{Orbit recovery.}
Using methods from algebraic geometry we can also give a criterion
for when a collection of moment polynomials allows for generic orbit recovery.
However, this criterion involves computing the dimension and degree of an algebraic variety. Such calculations can be done symbolically using a computer algebra system but not efficiently \cite[Appendix D]{bendory2020toward}.

To simplify the discussion we focus on the first and second order moments and
recall the strategy used in the proof of Theorem~\ref{thm:main}. Given a generic probability distribution
$\rho = \{p_g\}_{g \in G=D_{2L}}$ and a generic vector $x \in V= \R^L$, we proved that the following system
of bilinear equations in the $3L-1$ unknowns $x', p'_g$ has at most $2L= |G|$ solutions
\begin{equation} \label{eq.linearsystem}
\begin{array}{ccc}\sum_g p_g gx - p'_g gx' & =&  0,\\
  \sum_g (p_g (gx)^T gx - p'_g (gx')^T gx')  & = & 0.
  \end{array}
\end{equation}
(Note that the number of unknowns is $3L-1$ because $\sum_{g\in G} p'_g =1$ since $\rho'$ is a probability distribution and we can therefore express one of
the $p_g$ in terms of the other ones.)
The next proposition shows that our verification was equivalent to proving a statement
about an {\em incidence variety} associated to the group $G$ and vector space
$V = \R^L$. To formulate the result, we first establish notation for the action
of a finite group $G$ on a vector space $V$.
Let $I \subset (R(G) \times V)^2$ be the subvariety defined by the bilinear
equations~\eqref{eq.linearsystem}, where the $x,x', \{p_g\}, \{p'_g\}$ are all
considered variables. Since $\rho, \rho'$ are probability distributions
$\sum_{g\in G}p_g = \sum_{g \in G} p'_g$ so we can view this as a system of equations in $2(\dim V + |G| -1)$ variables.

In the language of algebraic geometry,  $I$ is called an {\em incidence variety}. 
The geometry of the incidence variety $I$ characterizes
when orbit and list recovery are possible.
\begin{prop} Let $G$ be a finite group acting on a vector space $V$, and let $I \subset (R(G) \times V)^2$ be an incidence defined in~\eqref{eq.linearsystem}.
  \begin{enumerate} 
  \item  If $\dim I = \dim V + |G|-1$, then for a generic signal $x$ and probability
    distribution $\rho = \{p_g\}_{g \in G}$, list recovery is possible from
    the first and second order moments $M^2(x, \rho)$.
  \item If $\dim I = \dim V + |G|$ and in addition $\deg I = |G|$, then for a 
    generic signal $x$ and probability distribution $\rho=\{p_g\}_{g\in G}$, orbit recovery is possible from the first and second orders moment
    $M^2(x,\rho)$.
	\end{enumerate}
\end{prop}
\begin{proof}
  Consider the projection $\pi \colon I \to R(G) \times V$ defined by
  $(x,\rho, x', \rho') \mapsto (x, \rho)$. If $\dim I = \dim V + R(G)-1$,
  then the generic fiber of $\pi$
  must be 0-dimensional. Hence, for a generic vector $x \in V$ and probability
  distribution $\rho \in R(G)$, there 
  can be at most a finite number of pairs $(x,\rho, x', \rho') \in I$.
  In other words,
  there are finite number of vectors $x'$
  such that there exist a
        distribution
	$\rho'$ with the property that $M^1(x,\rho) = M^1(x',\rho')$
        and $M^2(x,\rho) = M^2(x', \rho')$
        This proves part (i).

	Note that for each $g \in G$,  the set $X_g = \{(x,\rho ,gx,g\rho)|
        x \in V, \rho \in R(G)\}$ is a $\dim V + |G|-1$-dimensional subvariety of $I$, which is isomorphic to $R(G) \times V$. In particular,
        if $\dim I = \dim V + |G|$ then it must necessarily be an irreducible
        component of the variety $V$ in the sense of algebraic geometry.
        Hence, if $\dim I = \dim V + |G|-1$ then $I$
        has at least $|G|$ irreducible components.
and therefore its degree must be at least
	$|G|$. 
	Hence, if $\dim I = \dim |V| + |G|-1$ and $\deg I = |G|$, then $I$ has exactly $|G|$
	irreducible components and for generic $x, \rho$ there will be exactly
	$|G|$ pairs $(x,\rho, x', \rho') \in I$. Hence each $x'$ must necessarily equal 
	$gx$ for some $g \in G$.
        Therefore, the first and second order moments recover generic orbits $x$ in this case. 
\end{proof}
}

\section{Algorithms} \label{sec:algorithms}

In this section, we introduce  three algorithmic paradigms to estimate the signal $x$ from dihedral MRA observations $y_1,\ldots,y_n$ as in~\eqref{eq:mra_observations}.
We first introduce the three methods, and then compare them numerically in Section~\ref{sec:numerical_experiments}. 

\subsection{Group synchronization}
If the group elements $g_1,\ldots,g_n\in D_{2L}$ were known,  estimating the signal can be done by aligning the observations and averaging out the noise:
\begin{equation} \label{eqn:avg_sync}
{x_{\text{est}}} = \frac{1}{n}\sum_{i=1}^ng_i^{-1}y_i. 
\end{equation}
This motivates {synchronization} methods to estimate the unknown group elements from the observations. 
Synchronization starts by aligning all pairs of observations $y_i, y_j$, $i\neq j$, so that 
\begin{equation} \label{eqn:alignment}
y_i \approx g_{ij} y_j,
\end{equation}
for some group element $g_{ij}\in D_{2L}$. 
A standard alignment procedure is based on cross-correlating the observations. In more general groups, other common features can be harnessed; see for example~\cite{haming2010structure,singer2011three}. The relation~\eqref{eqn:alignment} is merely a proxy to $g_i \cdot x\approx g_{ij} g_j \cdot  x$, which in turn means that $g_i g_j^{-1} \approx g_{ij}$. At this stage, one reduces the MRA problem to the problem of \textit{group synchronization}~\cite{singer2011angular}, where we aim at estimating the unknown group elements $g_1,\ldots,g_n$ from a subset of their ratios $g_i g_j^{-1} $, often corrupted with noise. 

Early synchronization studies addressed the problem over compact groups, such as, finite groups, phases, and rotations. The common property of all synchronization cases over compact groups is that we can reduce them all to synchronization over rotations, or a subgroup of rotations, by using a faithful orthogonal representation~\cite{boumal2016nonconvex,  carlone2015initialization, perry2018message, tron2016survey}. 
Further generalizations extended   synchronization methods to  non-compact groups, and in particular to the Euclidean group, see e.g.,~\cite{ozyesil2018synchronization, rosen2019se,bendory2021compactification}.

Specifically for the dihedral MRA problem~\eqref{eq:mra}, we start by computing the cross-correlation between any observation $y_i$ and any other observation $y_j$ and its reflection $s y_j$. The maximal value indicates the best alignment as in~\eqref{eqn:alignment}. The resulting ratios $g_{ij} \approx g_i g_j^{-1} $ serve as an input for a standard spectral algorithm~\cite{singer2011angular}, which uses a rounding procedure onto the dihedral group, resulting in  estimates of the group elements {$\tilde g_1, \ldots, \tilde g_n$}. 
The orbit of the signal is then estimated by averaging  over the synchronized observations 
\begin{equation}
	{x_{\text{est}}} = \frac{1}{n}\sum_{i=1}^n{\tilde{g}}_i^{-1}y_i. 
\end{equation}

Unfortunately, in  low SNR environments the error of estimating the  ratios $g_i g_j^{-1}$, and thus of estimating $g_1,\ldots,g_n,$  grows rapidly~\cite{aguerrebere2016fundamental, robinson2004fundamental, robinson2009optimal, bendory2019multi}. Thus, in such regimes we consider techniques which aim to recover the signal $x$ directly, bypassing the estimation of the missing group elements~$\{g_i\}_{i=1}^n$. Next, we present two such methods, based on expectation-maximization and the method of moments. 

\subsection{Maximum likelihood estimation using expectation-maximization}
The log-likelihood function of~\eqref{eq:mra} is given by 
\begin{equation} \label{eq:likelihood}
	\ell(x,\rho)=\log p(y_1,\ldots,y_n; x,\rho) = \sum_{i=1}^N\log \sum_{j=1}^{2L}\rho[j]\frac{1}{(2\pi\sigma^2)^{L/2}}e^{-\frac{||y_i-g[j]\cdot x||^2}{2\sigma^2}},
\end{equation} 
where $g[1],\ldots,g[2L]$ are the elements of $D_{2L}$.
This is the standard likelihood function of a Gaussian mixture model, but all centers are connected through the orbit of $ D_{2L}$ acting on $x$. 
We wish to find the signal $x$ and distribution~$\rho$ that maximize~\eqref{eq:likelihood}. 
In {the} sequel, we assume no prior information on the signal and the distribution. If such information is available, then it is useful to consider the log-posterior distribution  $\log p(x,\rho|y_1,\ldots,y_n)$, which is equal to the log-likelihood plus the log of the prior terms.

To maximize the likelihood function, we devise an expectation-maximization (EM) algorithm~\cite{dempster1977maximum}. 
The EM algorithm has been successfully applied to other MRA setups~\cite{bendory2017bispectrum,abbe2018multireference,ma2019heterogeneous,janco2021accelerated} as well as for cryo-EM~\cite{scheres2012relion,sigworth1998maximum,bendory2020single}.
Although  EM is not guaranteed to achieve the maximum of the non-convex likelihood function~\eqref{eq:likelihood}, it is guaranteed that each EM iteration does not reduce the likelihood. 
{ In addition, for the general discrete MRA model, it was shown  that at low noise, this landscape is ``benign'', namely,  there are no spurious local optima (besides the maximum likelihood) and only strict saddle points. At high noise, this landscape
may develop spurious local optima, depending on the specific group. 
In addition, it was shown that the likelihood landscape is locally convex~\cite{fan2020likelihood}. 
 }

EM is an iterative algorithm, and each step consists of two steps. In the first step, called the E-step, the expectation of the complete likelihood (namely, the joint likelihood of $x,\rho$ and the group elements)  is computed. The expectation is taken with respect to the group elements (i.e., the nuisance variables),  given the current estimates of the signal $x_t$ and the distribution $\rho_t$:
	\begin{equation}
		\begin{split}
			Q(x,\rho |x_t,\rho_t) &= \E \left\{\log p(y_1,\ldots,y_n,g_1,\ldots,g_n; x, \rho) \right\} \\
			& = \sum_{i=1}^{n}\E\left\{-\frac{1}{2\sigma^2}\|y_i - g_i\cdot x\|^2+\log\rho[g_i]\right\} + \text{constant}\\
			& = \sum_{i=1}^{N}\sum_{j=1}^{2L}w_{i,j}\left\{-\frac{1}{2\sigma^2}\|y_i - g[j]\cdot x\|^2+\log\rho[j]\right\} + \text{constant},
		\end{split}
	\end{equation} 
where 
\begin{equation} \label{eq:em_weights}
	w_{i,j} =  \frac{\rho_t[j]e^{\frac{-1}{2\sigma^2}\|y_i-g[j] \cdot x_t\|^2 }}{\sum_{j=1}^{2L}\rho_t[j]e^{\frac{-1}{2\sigma^2}\|y_i-g[j]\cdot x_t\|^2 }}.
\end{equation}
The  second step, called M-step, maximizes $Q$ with respect to $x$ and $\rho$. In our case, the update step reads:
\begin{equation} \label{eq:Mstep}
	\begin{split}
	x_{t+1} = \frac{1}{n} \sum_{i=1}^N\sum_{j=1}^{2L}w_{i,j}g^{-1}[j]y_i	\\
	\rho_{t+1}[j] = \frac{\sum_{i=1}^n w_{i,j}}{\sum_{i=1}^n \sum_{j=1}^{2L}w_{i,j}}.
\end{split}
\end{equation}
If prior information is available (and thus the EM tries to maximize the posterior distribution rather than the likelihood), then it will act as a regularizer on the solution of the M-step. 
The EM algorithm iterates between computing the weights~\eqref{eq:em_weights} and updating the parameters~\eqref{eq:Mstep} until a stopping criterion is met.

\subsection{The method of moments}
The idea {behind} the method of moments is finding a pair $(x,\rho)$ whose moments match the empirical moments of the observations. 
In particular, according to Theorem~\ref{eq:mra}, only the first two moments are required to  uniquely characterize the orbit of  generic $x$ and $\rho$. 
The empirical moments can be computed from the data simply by averaging:
\begin{equation}
	\begin{split}
	{{M}^1_{\text{est}}} &= \frac{1}{n}\sum_{i=1}^n y_i, \\
	{{M}^2_{\text{est}}} &= \frac{1}{n}\sum_{i=1}^n y_iy_i^T.		
	\end{split}
\end{equation}
By the law of large numbers, and using Lemma~\ref{lem:moments}, for large $n$ we have 
\begin{equation}
	\begin{split}
	{{M}^1_{\text{est}}}\approx M^1(x,\rho) = C_{{ p}}x  + C_{{ q}} sx, \\
			{{M}^2_{\text{est}}} \approx M^2(x,\rho) = C_x D_{{ p}} C_x^T + C_{sx} D_{{ q}} C_{sx}^T,
	\end{split}
\end{equation}
where $C_z\in\R^{L\times L}$ is a circulant matrix generated by $z\in \R^L$, and  $D_z\in\R^{L\times L}$ is a diagonal matrix whose entries are $z$. 
As $n\to\infty$, $	{{M}^1_{\text{est}}}  {\to} M^1(x,\rho)$ and ${{M}^2_{\text{est}}} {\to} M^2(x,\rho)$ almost surely.

{A common practice is  to estimate  $(x,\rho)$ from ${{M}^1_{\text{est}}}$ and ${{M}^2_{\text{est}}}$} by minimizing a non-convex least squares objective:
\begin{equation} \label{eqn:LS}
\min _{\tilde{x} \in \mathbb{R}^{L}, [{p},{q}] \in \Delta^{2L}}
\norm{{{M}^2_{\text{est}}} - C_{\tilde{x}} D_{{ p}} C_{\tilde{x}}^T - C_{s\tilde{x}} D_{{ q}} C_{s\tilde{x}}^T}_{\mathrm{F}}^{2}
+\lambda
\norm{{{M}^1_{\text{est}}} - C_{{ p}}\tilde{x}  - C_{{ q}} s\tilde{x}}_{2}^{2}. 
\end{equation}
The solution of~\eqref{eqn:LS} is the method of moments estimator. While {the objective function~\eqref{eqn:LS}} is non-convex, it seems to provide accurate estimates in many cases. {In the low SNR regime, the method of moments is tightly connected to the maximum 
{likelihood} estimator. Specifically, in this regime likelihood optimization  reduces to a sequence of least squares optimization problems that match  moments~\cite{katsevich2020likelihood,fan2021maximum}. Since we use only two moments, the method of moments~\eqref{eqn:LS} can be interpreted as an approximation of the maximum likelihood estimator.}

\subsection{Numerical experiments}
\label{sec:numerical_experiments}
This section compares numerically the algorithmic methods discussed above: synchronization, expectation-maximization, and the method of moments. We define signal-to-noise ratio (SNR) as $\norm{x}^2/(L\sigma^2)$. To account for the group symmetry, we define  relative error as 
\begin{equation}
\text{relative error = } \min_{g\in D_{2L}}\frac{\|g\cdot {x_{\text{est}}}  - x\|}{\|x\|},
\end{equation}
where ${x_{\text{est}}} $ is the signal estimate.{ The entries of the  ground-truth $x$ of length $L=10$ were drawn i.i.d.\ from a normal distribution  with mean zero and variance one, and the distribution $\rho$ was uniformly sampled from the simplex $\Delta_{2L}$.  }

We consider two regimes: (i) a relatively small number of observations ($n=1000$) and moderate SNR levels,  and (ii) large $n$ and low SNR. The code to reproduce all experiments is publicly available at~\url{https://github.com/nirsharon/DihedralMRA}. The results below represent the average over $50$ trials. We initialized the EM algorithm from a single random point and halted it when the difference of the likelihood between two consecutive iterations dropped below $10^{-4}$, or after a maximum of $400$ iterations. For the method of moments, we minimized~\eqref{eqn:LS} using the trust-regions method; 
we initialized the optimization algorithm from $10$ different random initial guesses and chose the one that yields the least value of the cost~\eqref{eqn:LS}. The number of trust-regions iterations  was limited to $200$.

\paragraph{Moderate SNR regime.}
We begin with a noise regime where the synchronization approach presents a viable alternative to EM and the method of moments. Figure~\ref{fig:mod_noise} shows the  relative error of the three methods as a function of the SNR  with $n=1000$ observations. The method of moments shows inferior results compared to synchronization and EM since the empirical moments do not approximate the population moments accurately enough for such a small number of observations. For high SNR, the performance of  synchronization and EM are comparable.
{
The synchronization behavior is thus compatible with current knowledge about the synchronization problem and the spectral algorithm specifically, see, e.g.,~\cite{cucuringu2016sync, gao2019geometry}. 
}
 However, as the SNR drops, synchronization fails to estimate the group elements accurately, while both the method of moments and EM  present consistent error rates. 
{
This phenomenon agrees with theoretical findings regarding alignment in the presence of high noise~\cite{aguerrebere2016fundamental} and synchronization when applied to such corrupted input data~\cite{singer2011three}.
}
As the SNR approaches $1$, when the signal and the noise are of the same order, the synchronization method introduces relative error close to~$1$, meaning it contributes no  information about the solution.

\begin{figure}
	\begin{center}
		\includegraphics[width = .5\textwidth]{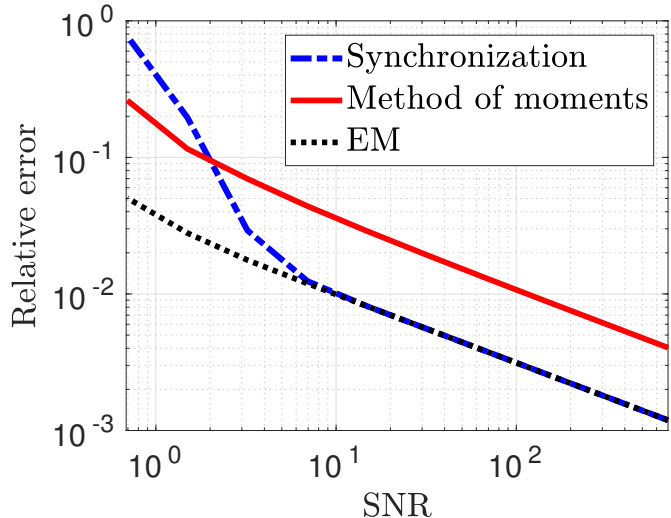}
		\caption{The relative error of the three methods under moderate SNR levels with $n=1000$ observations. As the SNR deteriorates, the synchronization method fails to estimate the group elements, and thus the signal, accurately.}
		\label{fig:mod_noise}
	\end{center}
\end{figure}

\paragraph{Low SNR.} 
We discard the synchronization algorithm in the low SNR regime as it cannot cope with  high noise levels, as demonstrated in Figure~\ref{fig:mod_noise}. In addition, since the first step of the synchronization method involves pairwise alignment, the synchronization input consists of $\mathcal{O}(n^2)$ group elements, and so the computational complexity of this method makes it impractical for as many as $n=10^5$ observations.

Figure~\ref{subfig:error} shows  relative errors as a function of SNR. The EM outperforms the method of moments for SNR values  above 1/10. For lower SNR levels,  the method of moments shows similar estimation rates.
{In the high SNR regime, the error curves of both methods scale as $\text{SNR}^{-1/2}$, namely as $\sigma$, which is the same estimation rate as if the group elements were known. {In particular, the numerical slope of the EM method is $-0.4999$ and the method of moments presents a numerical slope of $-0.5104$}. In the low SNR regime, however, the error curves scale as $\text{SNR}^{-1}\propto \sigma^2$. 
While this slope is expected for the method of moments that directly uses the first two moments (and thus its standard deviation is proportional to $\sigma^2$), the moments do not appear explicitly in the EM iterations. {Specifically, the numerical slopes for SNR values below $1/10$ were $-1.0561$ and $-1.1058$ for the method of moments and EM, respectively.} This rate implies that accurate estimation requires $n\gg\sigma^4$, corroborating our theoretical findings (Corollary~\ref{cor:sample_complexity}) that no algorithm can achieve better estimation rates in the low SNR regime. A similar phenomenon was observed by previous MRA studies~\cite{sigworth1998maximum, bendory2017bispectrum, abbe2018multireference, bendory2020super}. For the connection between EM and the method of moments in the low SNR regime, see~\cite{katsevich2020likelihood,fan2020likelihood,fan2021maximum}}. 	

Figure~\ref{subfig:time} presents the corresponding average runtime. The runtime  of EM increases as the SNR decreases, while the runtime of the  method of moments remains roughly constant. The reason for the growth in runtime is revealed in Figure~\ref{fig:EM_iters}, where we display the average number of EM iterations  as a function of  SNR. The figure  shows that the number of iterations is inversely proportional to the SNR. 
 
\begin{figure}
	\begin{center}
		\begin{subfigure}[ht]{0.45\textwidth}
			\centering
			\includegraphics[width=\textwidth]{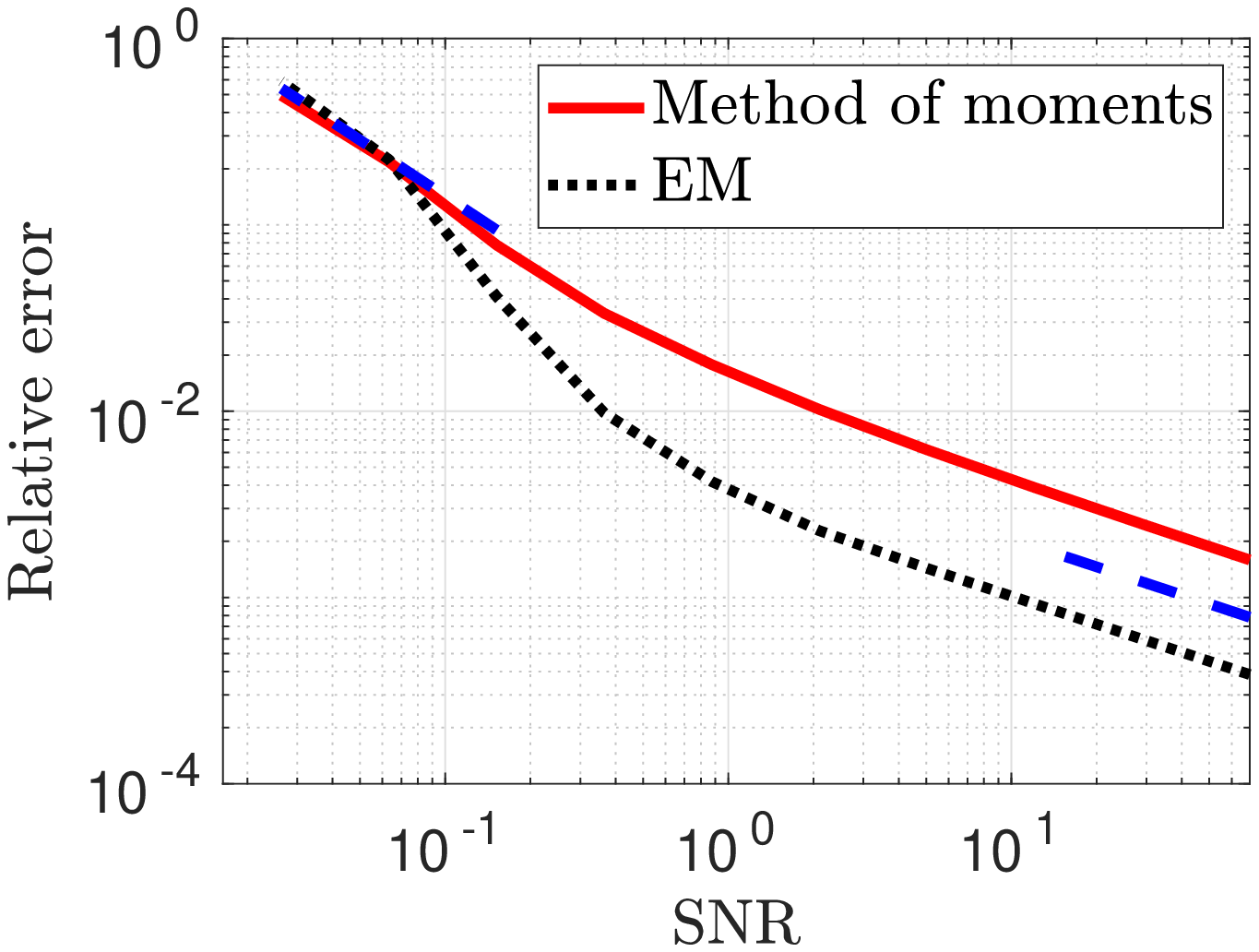}
			\caption{}
			\label{subfig:error}
		\end{subfigure}
		\quad
		\begin{subfigure}[ht]{0.45\textwidth}
			\centering
			\includegraphics[width=\textwidth]{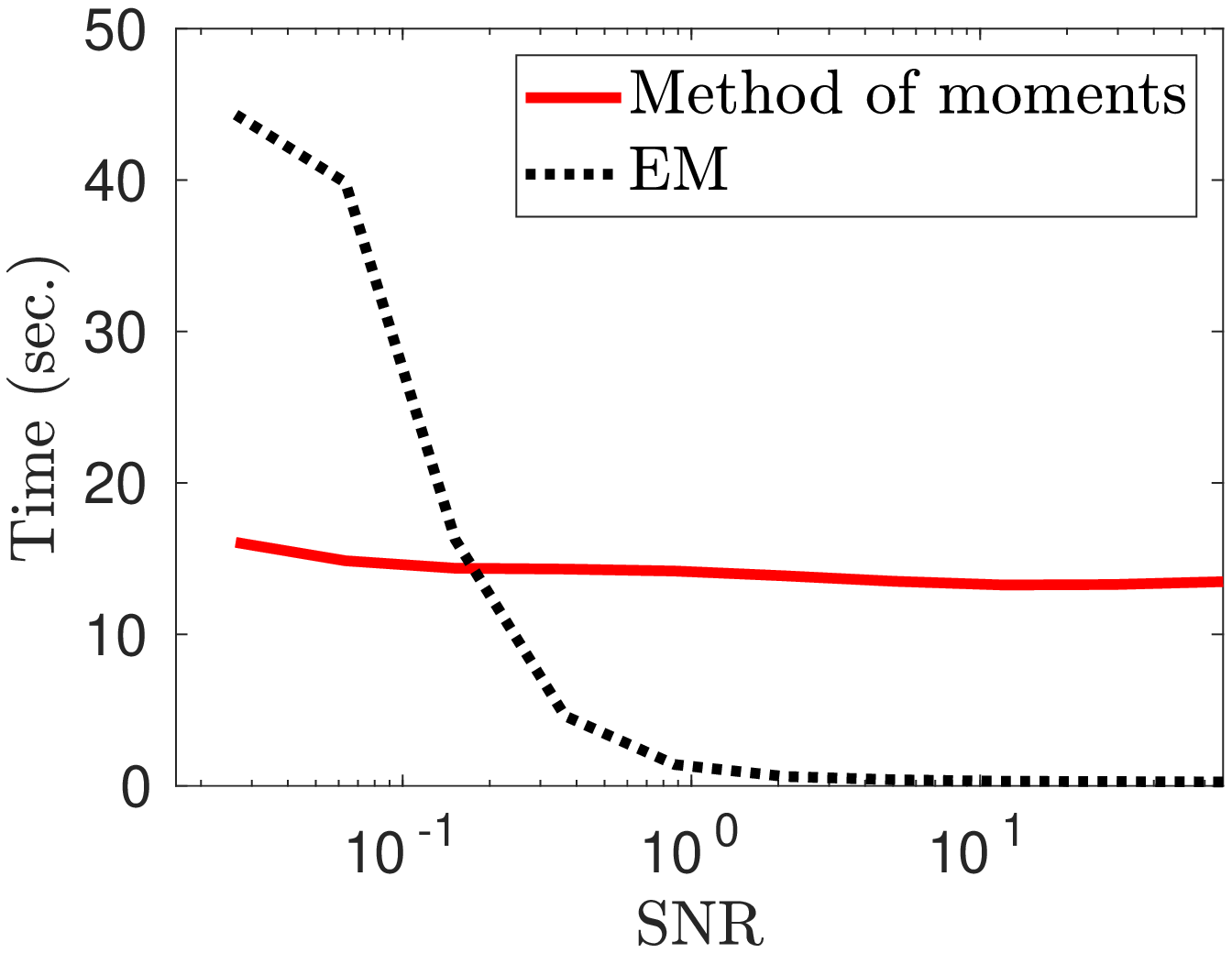}
			\caption{}
			\label{subfig:time}
		\end{subfigure}
		\caption{Relative error {(left panel)} and runtime {(right panel)} of the method of moments and EM as a function of the SNR for  $n=10^5$ observations. 
		{In the high SNR regime, the slope of the error curves (dashed blue line), for both methods,  scales as $\text{SNR}^{-1/2}\propto \sigma$, which is the same estimation rate as if the group elements were known. In the low SNR regime, however, the error curves scale as $\text{SNR}^{-1}\propto \sigma^2$, corroborating our  theoretical findings.} 	
		\label{fig:largeN_comparison}}
	\end{center}
\end{figure}

\begin{figure}
	\begin{center}
		\includegraphics[width = .45\textwidth]{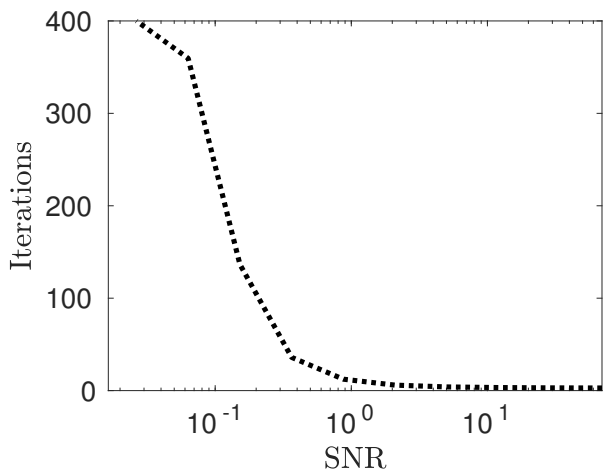}
		\caption{The average number of EM iterations as a function of the SNR. The maximum number of iterations is set to $400$.}
		\label{fig:EM_iters}
	\end{center}
\end{figure}

\section*{Acknowledgment}
W.L. and N.S. are partially supported by BSF grant no. 2018230. 
T.B. and D.E. are partially supported by BSF grant no. 2020159.
T.B. and N.S are partially supported by the NSF-BSF award 2019752. 
T.B. is also supported in part by the
 ISF grant no. 1924/21.
 D.E. is supported by Simons Collaboration grant 708560. W.L. is partially supported by NSF award IIS-1837992.

\bibliographystyle{plain}

\end{document}